\documentclass[aap]{imsart}
\usepackage{amsthm,amsmath}

\usepackage{amssymb}
\usepackage{amsfonts}

\usepackage{amsmath} \usepackage{amsfonts}
\usepackage{amssymb} \usepackage{latexsym} \usepackage{enumerate}
 \usepackage{mathrsfs}
\usepackage[numbers]{natbib}
\usepackage{todonotes}

\usepackage[utf8]{inputenc}

\RequirePackage[colorlinks,citecolor=blue,urlcolor=blue]{hyperref}
\hypersetup{pageanchor=false}

\newtheorem{Theorem}{Theorem}[section] %
\newtheorem{Lemma}[Theorem]{Lemma} %

\newtheorem{thm}{Theorem}[section]

\newtheorem{rem}[Theorem]{Remark}

\newcommand{\set}{:=}

\DeclareMathOperator*{\Law}{Law}

\renewcommand{\mathcal}{\mathscr}
\newcommand{\cA}{\mathcal{A}}

\newcommand{\cE}{\mathcal{E}}
\newcommand{\cF}{\mathcal{F}}

\newcommand{\cT}{\mathcal{T}}

\newcommand{\DD}{\mathbb{D}}
\newcommand{\EE}{\mathbb{E}}

\newcommand{\HH}{\mathbb{H}}

\newcommand{\PP}{\mathbb{P}}
\newcommand{\QQ}{\mathbb{Q}}
\newcommand{\RR}{\mathbb{R}}
\newcommand{\TT}{\mathbb{T}}

\newcommand{\VV}{\mathbb{V}}

\renewcommand{\epsilon}{\varepsilon}

\newcommand{\supp}{\operatorname{\mathrm{supp}}}

\newcommand{\conv}{\operatorname{\mathrm{conv}}}

\newcommand{\cbr}[1]{\left\{ #1 \right\}}

\renewcommand{\P}{\PP}

\newcommand{\E}{\EE}

\begin{document}
\begin{frontmatter}
\title{Super-Replication with Fixed Transaction Costs}
 \runtitle{Super-Replication with Fixed Transaction Costs}
 \author{Peter Bank$^*$ and Yan Dolinsky$^\dagger$ \corref{}}

\address{P.Bank:
 Department of Mathematics, TU Berlin.
 e.mail: bank@math.tu-berlin.de\\
Y.Dolinsky Department of Statistics, Hebrew University
  and School of Mathematical Sciences, Monash University.
 e.mail:yan.dolinsky@mail.huji.ac.il}
 \affiliation{Hebrew University$^\dagger$ and Monash University$^\dagger$, and TU Berlin$^*$}
\runauthor{P.Bank and Y.Dolinsky}
\thankstext{}{Both authors are grateful to the Einstein Foundation for the financial support through its
research project on “Game options and markets with frictions”. The author YD is also partially supported by
Marie--Curie Career Integration Grant, no. 618235 and ISF Grant 160/17}
 \date{\today}
\begin{abstract}\noindent
 We study super-replication of contingent claims in markets with
  fixed transaction costs.  This can be viewed as a stochastic impulse
  control problem with a terminal state constraint. The first result
  in this paper reveals that in reasonable continuous time financial
  market models the super-replication price is prohibitively costly
  and leads to trivial buy-and-hold strategies.  Our second result
  derives nontrivial scaling limits of super-replication prices for
  binomial models with small fixed costs.
\end{abstract}
\begin{keyword}[class=MSC]
\kwd[Primary ]{91G10, 91G20}
\end{keyword}
\begin{keyword}
\kwd{binomial models} \kwd{conditional full support} \kwd{fixed transaction costs} \kwd{super--replication}
\end{keyword}
\end{frontmatter}
\section{Introduction}\label{sec:1}

This paper deals with super-replication of European options in a
market where trading incurs fixed transaction costs.  Most papers
dealing with fixed transaction costs explore the problem of optimal
portfolio choice (see for instance, \cite{AJS}, \cite{EH}, \cite{Ko},
\cite{LMW}, \cite{MP} and \cite{OP}). Much fewer papers (see
\cite{JKN} and \cite{LT}) discuss no arbitrage criteria for fixed
transaction costs and, to the best of our knowledge, the problem of
super--replicating a contingent claim with fixed costs has not been
considered in the literaure before.

By contrast, for the case of proportional transaction costs, the topic
of super-replication is widely studied.  In \cite{DS} it was
conjectured that, in the Black-Scholes model with proportional
transaction costs, the cheapest way to super-replicate a call option
is to buy one unit of stock right at the start and hold it until
maturity. This conjecture was proved by many authors (see e.g.,
\cite{CPT}, \cite{GRS}, \cite{JLR}, \cite{LS}, \cite{SSC} and for game
options in \cite{D1}).  A natural way to overcome this negative result
was proposed by Kusuoka in~\cite{K}.  He considered scaling limits of
the classical Cox-Ross-Rubinstein model of a complete binomial market
and showed that, when transaction costs are also rescaled properly,
super-replication prices converge to what is now known as a
$G$-expectation in the sense of Peng (\cite{P}).

The present paper is a first step in the development of the above
theory for the fixed transaction costs case.  The setup of fixed
transaction costs corresponds to the case where any (nonzero)
transaction incurs a fixed cost of $\kappa>0$, regardless of the
trading volume.  Clearly, this leads to discontinuous, non-convex
wealth dynamics which induce a stochastic impulse control problem with
a novel terminal state constraint. In particular, convex duality
methods, which played a key role in the theory of proportional
transaction costs (or their convex generalizations), are not available
here.

As a first result, we show in Theorem~\ref{thm:1} that, in a
continuous time financial market with a risky asset exhibiting
conditionally full support (see \cite{GRS}), the cheapest way
to super-replicate a convex option is again to apply a trivial buy-and-hold
strategy.  Hence, Theorem~\ref{thm:1} can be viewed as an analog for
fixed costs of the result in \cite{GRS} which was obtained for the
case of proportional transaction costs. By contrast to the classical
duality used in~\cite{GRS}, our proof uses the impulse control
structure directly.

The second result in the present paper deals with the limiting
behavior of super-replication prices in the Cox-Ross-Rubinstein
binomial models of \cite{CRR}.  Specifically, we consider a sequence
of binomial models with constant volatility and study the asymptotic
behavior of the super-replication prices for convex payoffs when the
time step goes to zero and the fixed transaction costs are scaled
linearly as a function of the time step. In Theorem~\ref{thm:2} we
characterize the scaling limit as a stochastic volatility control
problem defined on Wiener space.

Our proof relies heavily on the fact that the payoff of the European
option is a convex function of the risky asset.  Under this assumption
we derive a non-standard dual representation for super-replication
prices in the binomial models.  This representation allow us to obtain
the limit behavior of the super-replication prices by modifying ideas
from \cite{K}.  We emphasize that without the convexity condition on
the payoff the analysis is more complicated and remains an open
question.

Closely related is the topic of \emph{approximate hedging} which deals
with the construction of portfolio strategies with terminal wealth
close to the payoff of the derivative security. Approximate hedging in
the context of market frictions is going back to the pioneering work
of Leland \cite{L} who considers a Black--Scholes model with vanishing
proportional transaction costs.  This approach is studied
rigorously and extended (beyond Black--Scholes and beyond vanishing
proportional transaction costs) in \cite{CF,EL,Geiss,KS,L1,Pa}. The
triviality of super-replication prices established in our
Theorem~\ref{thm:1} can also be viewed as a motivation for the study
of approximate hedging in the fixed transaction costs setup.

The paper is organized as follows.  In Section~\ref{sec:2} we
formalize the super-replication problem with fixed
costs. Section~\ref{sec:buy-and-hold} shows that, in models with
conditional full support, trivial buy-and-hold strategies yield
optimal super-replications of convex payoffs.  In
Section~\ref{sec:binomial} we give the scaling limit of
super-replication prices with small fixed costs. The proof of this
result is prepared by a dual representation for our super-replication
prices discussed in Section~\ref{sec:duality} and accomplished in
Sections~\ref{sec:lowerbound} and~\ref{sec:upperbound} by using tools
from weak convergence of stochastic processes to analyze the
asymptotic behavior of the dual terms.

\section{Superreplication with fixed transaction costs}\label{sec:2}

Let $(\Omega,\cF,(\cF_t),\P)$ be a filtered probability space with a
progressively measurable process $S>0$ which we take to describe the
price evolution of some financial asset with initial price
$S_0=s_0>0$. The asset is traded at strictly positive fixed costs
$\kappa>0$ per transaction and so an investor with a bank account (that
for simplicity bears no interest) can change her position only finitely
often. We take $T=1$ to be the investor's time horizon and so the
times of intervention are given by a family of stopping times
$\TT=(\tau_i)_{i=1,2,\dots}$ such that
$$
0 = \tau_0 \leq \tau_1 \leq \tau_2 \leq \dots \leq T = 1  \text{ with } \tau_i <
\tau_{i+1} \text{ on } \cbr{\tau_i<1}.
$$
Let us denote by $\cT$ the class of all such families $\TT$ for which
the number of interventions by time $T=1$ is finite almost surely:
$$
N(\TT) \set \sup\cbr{i =0,1,\dots \;:\; \tau_i<1} < \infty \text{
  $\P$-a.s.}.
$$
Notice that for simplicity we do not count a possible initial
intervention at time $\tau_0=0$.

Assume our investor seeks to hedge an option with $\cF_1$-measurable
payoff $F\geq 0$ at maturity $T=1$ by an investment strategy
$(\TT,\HH)$ where $\HH=(h_i)_{i=0,1,\dots}$ describes the
$\cF_{\tau_i}$-measurable number of assets $h_i$ to be held,
respectively, over each period $(\tau_i,\tau_{i+1}]$,
$i=0,1,\dots$. Keeping in mind the fixed transaction costs $\kappa>0$
and the free trade at time 0, the investor's gains from trading will
by time $t \leq 1$ have accrued to
$$
G^{\kappa}(\TT,\HH)_t \set   \sum_{i=0,1,\dots}h_i(S_{\tau_{i+1} \wedge t}-S_{\tau_i \wedge t})
 - \kappa \sup\cbr{i=0,1,\dots \;:\; \tau_i < t}.
$$
To rule out the possibility of doubling strategies, the investor can
only use admissible strategies from the set
$$
\cA \set \cbr{(\TT,\HH) \;:\; G^{\kappa}(\TT,\HH) \text{ bounded from below by
    a constant $\P$-a.s.}}.
$$
The option's  super-replication price is then given by
$$
\VV^\kappa(F) \set \inf \cbr{x \in \RR \;:\; x+G^{\kappa}(\TT,\HH)_1
  \geq F \text{ $\P$-a.s.  for some } (\TT,\HH) \in \cA}.
$$
Determining this super-replication price amounts to solving an impulse
control problem with terminal state constraint, a task which cannot be
carried out explicitly without further assumptions. We will show
however that for convex payoffs it can be computed in models with
conditional full support (Section~\ref{sec:buy-and-hold}). At the
other end of the modeling spectrum, we consider binomial models
converging to a Black-Scholes dynamic, for which we compute the
scaling limit for suitably scaled fixed costs
(Section~\ref{sec:binomial}).
\begin{rem}
  In the frictionless case $\kappa=0$ with continuous stock prices,
  the above super-replication price is the classical one even given
  the constraint to an almost surely finite number of trades. This
  follows readily from the fact that the wealth process of any
  continuous-time trading strategy can be approximated uniformly (in
  time and almost all scenarios) by piecewise constant (admissible)
  trading strategies (Lemma A.3 in \cite{LS}).
\end{rem}

\section{Buy-and-hold with conditional full support}\label{sec:buy-and-hold}

In this section, we consider a continuous model
$S=(S_t)_{t \in [0,1]}$ exhibiting \emph{conditional full support} as
discussed by, e.g., \cite{GRS}:
\begin{equation}\label{eq:10}
\supp \mathbb P[S|_{[t,1]} \in \cdot |\mathcal F_t]=C^{+}_{S_t}[t,1]
\mbox{  $\P$-a.s. for any } t \in [0,1],
\end{equation}
where, for $y \geq 0$, $C^{+}_y[t,1]$ denotes the space of all continuous
paths $[t,1]\rightarrow\mathbb{R}_{+}$ starting in $y$ at time $t$.

\begin{thm}\label{thm:1}
  For any financial model exhibiting conditional full support in the
  sense of~\eqref{eq:10}, the super-replication price with fixed
  transaction costs $\kappa>0$ of any convex payoff $F=f(S_1)$ with
  $f:[0,\infty) \to \RR$ continuous and convex is
  \begin{equation}
  \VV^{\kappa}(f(S_1)) = f(0)+s_0 f'(\infty) \text{ where } f'(\infty) \set
  \sup_{s>0} f'(s).\label{eq:6}
  \end{equation}
  In case $f'(\infty)<\infty$, a super-hedge with initial capital
  $\VV^{\kappa}(f(S_1))$ is to buy $h_0\set f'(\infty)$ units of the asset
  at time $\tau_0=0$ and hold these until $T=1$.
\end{thm}

\begin{proof}
That the right-hand side of~\eqref{eq:6} is sufficient for
super-replication is trivial if $f'(\infty)=\infty$. If
$f'(\infty)<\infty$, we can consider the described buy-and-hold
strategy which yields
$$
G^{\kappa}(\TT,\HH)_1=f'(\infty)(S_1-S_0) \geq
f'_-(S_1)S_1-f'(\infty)S_0 \geq f(S_1)-f(0)-f'(\infty)S_0
$$
where both estimates are due to the convexity of $f$. This shows that
$x_0 \set f(0)+f'(\infty)S_0$ is enough to super-replicate $F=f(S_1)$.

Now consider $x<x_0$ and take a strategy $(\TT,\HH)$ with gains
process $G \set G^{\kappa}(\TT,\HH)$ such that $x+G_1 \geq
F=f(S_1)$. We will show that such a strategy cannot be
admissible. Specifically, with $\beta<f'(\infty)$ such that
$x=f(0)+\beta S_0$, we will argue that
$$
A_n \set \cbr{\tau_n<1, \, S_{\tau_n}<2/\delta, \,
  x+G_{\tau_n}<f(0)+\beta S_{\tau_n}-n \kappa/2}
$$
has positive probability for all $n=1,2,\dots$, where,
$\delta \in (0,1/s_0)$ is chosen small enough to ensure
$$
f(0)+\beta s - \kappa < f(s) \text{ for all } s<\delta \text{ and all
} s>1/\delta.
$$
Such a choice of $\delta$ is possible since $f$ is continuous at zero
and convex on $[0,\infty)$ with $f'(\infty)>\beta$. Since $\kappa>0$,
it then follows that $x+G_{\tau_n}<f(0)+2\beta/\delta-n\kappa/2$ on
the set $A_n$ with positive probability, $n=1,2,\dots$, and so
$G=G^\kappa(\TT,\HH)$ is not bounded from below by a constant and
$(\TT,\HH)$ cannot be admissible.

We will prove $\P[A_n]>0$, $n=0,1,\dots$, by induction. By our choices
of $\delta<1/s_0$ and of $\beta$, we even have $\P[A_0]=1$. Now
assume, by way of contradiction, that $\P[A_n]>0$, but $\P[A_{n+1}]=0$
for some $n$. Observe that on $A_n \cap \cbr{\tau_{n+1}<1}$ we can
estimate
\begin{align}\nonumber
  x+G_{\tau_{n+1}}& =
                    x+G_{\tau_n}+\beta(S_{\tau_{n+1}}-S_{\tau_n})+(h_n-\beta)(S_{\tau_{n+1}}-S_{\tau_n}) -\kappa\\
\label{eq:110}
&< f(0)+\beta S_{\tau_{n+1}}-(n+1)
  \kappa/2+(h_n-\beta)(S_{\tau_{n+1}}-S_{\tau_n}) - \kappa/2.
\end{align}
Hence, $A_{n+1}$ contains the set $A_n \cap \cbr{\tau_{n+1}<1} \cap B_{n}$
where
$$
B_{n} \set \cbr{\sup_{\tau_n \leq t \leq 1} S_t<2/\delta, \
  \sup_{\tau_n\leq t \leq 1}\cbr{
  (h_n-\beta)(S_{t}-S_{\tau_n})}\leq \kappa/2 }.
$$
Notice that $\tau_{n+1}=1$ must hold almost surely on $A_n \cap
B_n$ since, with $\P[A_{n+1}]=0$, we also have
$$
0  = \P[A_n \cap \cbr{\tau_{n+1}<1} \cap B_n]=\P[A_n
\cap B_n]-\P[A_n \cap \cbr{\tau_{n+1}=1} \cap B_n].
$$

Now, $A_n \cap B_n$ contains $A_n \cap \cbr{h_n \geq \beta} \cap C_n$ where
$$
C_n \set B_n \cap \cbr{S_1 \geq S_{\tau_n} \vee 1/\delta}.
$$
On $A_n \cap \cbr{h_n \geq \beta} \cap C_n$, however, the
super-replication property is violated since, on this set, we have
$\tau_{n+1}=1$ almost surely and estimate~\eqref{eq:110} gives
$$
x+G_1 = x+G_{\tau_{n+1}} < f(0)+\beta S_{1}-(n+1)\kappa/2<f(S_1)
$$
by choice of $\delta$ and definition of $C_n$. Hence, we deduce
$$
0 =\P[ A_n \cap \cbr{h_n \geq \beta} \cap C_n ]=\E\left[1_{A_n \cap \cbr{h_n \geq \beta}} \P[ C_n \;|\;\cF_{\tau_n}]\right].
$$
As the conditional full support property~\eqref{eq:10} holds also
at stopping times when it holds at deterministic times (see
Lemma~2.9 in~\cite{GRS}), we have $\P[C_n \;|\; \cF_{\tau_n}]>0$ almost
surely on $A_n \cap \cbr{h_n \geq \beta}$. So
the above identity yields that in fact $\P[A_n \cap \cbr{h_n \geq \beta}]=0$.

Similarly we will argue next that $\P[A_n \cap \cbr{h_n < \beta}]=0$
so that in conjunction with $\P[A_n \cap \cbr{h_n \geq \beta}]=0$ we
arrive at the contradiction $\P[A_n]=0$, completing our proof. Thus,
let us first observe that $A_n \cap B_n$ contains
$A_n \cap \cbr{h_n < \beta} \cap \tilde{C}_n$ where
$$
\tilde{C}_n \set B_n \cap \cbr{S_1 \leq S_{\tau_n} \wedge \delta}.
$$
Up to a $\P$-null set, however, we still have
$A_n \cap \cbr{h_n < \beta} \cap \tilde{C}_n \subset
\cbr{\tau_{n+1}=1}$ and the super-replication property is again
violated since, on this set, estimate~\eqref{eq:110} gives
$$
x+G_1 = x+G_{\tau_{n+1}} < f(0)+\beta S_{1}-(n+1)\kappa/2<f(S_1)
$$
by choice of $\delta$ and definition of $\tilde{C}_n$. Observing that
also $\P[\tilde{C}_n \;|\;\cF_{\tau_n}]>0$ almost surely on
$A_n \cap \cbr{h_n < \beta} \in \cF_{\tau_n}$ allows us to deduce by
the same reasoning as used for $C_n$ that indeed
$\P[A_n \cap \cbr{h_n < \beta}]=0$.
\end{proof}

\begin{rem}
  If we restrict ourselves to admissible strategies, the conditional
  full support property~\eqref{eq:10} guarantees absence of arbitrage
  (as it also does for proportional transaction costs; see
  \cite{G,GRS}). Indeed, assume that for a trading strategy
  $(\TT,\HH)$ we have $G^{\kappa}(\TT,\HH)_1\geq 0$ $\mathbb
  P$-a.s. and $\mathbb P(G^{\kappa}(\TT,\HH)_1> 0)>0$. Then, similarly
  to the above proof, we can argue by induction that for any
  $n=1,2,\dots$,
  $\mathbb P (\tau_n<1, G^{\kappa}(\TT,\HH)_{\tau_n}<-n\kappa/2)>0$
  and, thus, the gain process $G^{\kappa}(\TT,\HH)_t$, $t\geq 0$ is not
  uniformly bounded from below. So $(\TT,\HH)$ is not admissible. For
  more refined no arbitrage criteria we refer to \cite{JKN, LT}.
\end{rem}
\section{Scaling limit of binomial superreplication prices}\label{sec:binomial}

In this section we consider binomial Cox-Ross-Rubinstein models with
fixed transaction costs and describe the scaling limit of
superreplication prices for convex claims. To wit, we let
$\Omega=\cbr{-1,+1}^{\mathbb N_0}$, put $\zeta_i(\omega)=\omega_i$ for
$\omega = (\omega_0,\omega_1,\dots) \in \Omega$ and let $\P$ be the
measure under which the $\zeta_i$ are i.i.d. with
$\P[\zeta_i=1]=1/2$. The $n$-period binomial price process can now be
specified as
\begin{equation}\label{eq:112}
S^{(n)}_{t}=s_0\exp\left({ \frac{\sigma}{\sqrt n}}\sum_{i=1}^{[nt]}
  \zeta_i\right), \quad  t\in [0,T],
\end{equation}
and the underlying filtration $(\cF^{(n)}_t)$ is the one generated by
$S^{(n)}$.

Obviously, when considered under their respective equivalent
martingale measures $\P^{(n)}\approx \P$, these Cox-Ross-Rubinstein
models $S^{(n)}$, $n=1,2,\dots$, converge to a Black-Scholes model
with constant volatility $\sigma>0$. In light of Theorem~\ref{thm:1},
it is clear that in order to get a non-trivial limit for the
corresponding super-replication prices with fixed transaction costs,
one has to rescale the fixed costs suitably. Our next result shows
that the correct scaling is of the order $1/n$ and it identifies the
resulting scaling limit as a $G$-expectation with penalty involving
stochastic volatility models. These are specified as martingale
exponentials
\begin{equation}\label{eq:66}
S^{(\nu)}_t=s_0\exp\left(\int_{0}^t \nu_u dW_u-
\frac{1}{2}\int_{0}^t \nu^2_u du \right),
\quad t \in [0,1],
\end{equation}
where $W$ is a standard Brownian motion on some complete probability
space $(\Omega^W, \mathcal{F}^{W}, \mathbb{P}^{W})$ and where $\nu$ is
taken from the set $\mathcal A^W$ of all bounded, real-valued
processes $\nu \geq \sigma$ on this space which are progressively
measurable with respect to the augmented filtration
$(\cF^W_t)_{t \in [0,1]}$ generated by $W$.

\begin{thm}\label{thm:2}
  For a convex payoff $F=f(S_1)$ with continuous, convex $f:[0,\infty)
  \to \RR$ with polynomial growth, the scaling limit of
  superreplication prices in the binomial models~\eqref{eq:112} with
  fixed costs $\kappa/n$, $n=1,2,\dots$, is
 \begin{equation}
   \label{eq:5}
   \lim_{n \to \infty} \VV^{\kappa/n}(f(S^{(n)}_1))=\inf_{\sigma \leq \nu \in \cA^W}
   \E^W\left[f(S^{(\nu)}_1)+\kappa \int_0^1 g(\nu_t^2/\sigma^2)\,dt\right]
\end{equation}
where $g:[1,\infty) \to (0,1]$ is the linear interpolation supported by
$g(n)=1/n$, $n=1,2,\dots$ and where the infimum is taken over all the
probability spaces and volatility processes $\nu \geq \sigma$
described above.
\end{thm}

The proof of Theorem \ref{thm:2} is prepared by a duality result for
super-replication with fixed costs presented in
Section~\ref{sec:duality}. Section~\ref{sec:lowerbound} then
establishes ``$\geq$'' and Section~\ref{sec:upperbound} proves
``$\leq$'' in~\eqref{eq:5}, completing the proof.

Let us explain the intuition behind the above result. As will also be
revealed by our proof below, the local volatility pattern $\nu$ can be
viewed as a continuous-time measure of trading activity. For this
pattern to attain the infimum in~\eqref{eq:5}, it has to trade off the
option price $\E^W\left[f(S^{(\nu)}_1\right]$ against the expected
costs $\E^W\left[\int_0^1 g(\nu_t^2/\sigma^2)\,dt\right]$. Indeed,
since $f$ is convex, the option price is increasing as a function
of the volatility pattern $\nu \geq \sigma$ and thus would be
minimized by $\nu \equiv \sigma$. This choice, however, incurs the
maximum penalty as $g$ is decreasing. This increased reference
volatility is reminiscent of Leland's frictional trading recipe which
suggests to use a delta hedging strategy with increased local
volatility for approximate hedges with vanishing proportional
transaction costs.

\subsection{Duality for binomial models with fixed transaction costs}
\label{sec:duality}

The starting point for the proof of Theorem~\ref{thm:2} is a form of
dual characterization of super-replication prices with fixed costs in
binomial models which works for the special case of convex payoff
profiles.

To specify this duality, let us fix $n \in \cbr{1,2,\dots}$ and
consider the class $\cT^{(n)}$ of systems
$\TT=\cbr{0=\tau_0 \leq \dots \leq \tau_n=1} \in \cT$ of
$(\cF^{(n)}_t)$-stopping times with values in $\cbr{0/n,1/n,\dots,1}$
such that if $\tau_{k+1}(\omega)<1$ then
\begin{align} \label{eq:111} \begin{split}
\xi(\omega) &\equiv +1 \text{ for all } i \in
  \cbr{n\tau_k(\omega)+1,\dots,n\tau_{k+1}(\omega)}\\
\text{or }
\xi(\omega)& \equiv -1 \text{ for all } i \in
  \cbr{n\tau_k(\omega)+1,\dots,n\tau_{k+1}(\omega)}.\end{split}
\end{align}
In other words, the stopping times $\tau_k \leq \tau_{k+1}$ are such
that $\tau_{k+1}(\omega)=1$ in scenarios $\omega$ where
$S^{(n)}(\omega)$ is not strictly increasing or strictly decreasing
between $\tau_k(\omega)$ and $\tau_{k+1}(\omega)$. Also, already at
time $\tau_k$ it is known by how many downward steps and how many
upward steps the next stop $\tau_{k+1}$ will be reached. In other
words, for suitable functions
$\phi^\downarrow_k,\phi^\uparrow_k:\RR_+^{k+1} \to \mathbb{N}$, the
number of these steps can be written in the form
$\phi^\downarrow_k(S^{(n)}_0,\dots,S^{(n)}_{\tau_k})$ and
$\phi^\uparrow_k(S^{(n)}_0,\dots,S^{(n)}_{\tau_k})$, respectively, for
each $k=0,1,\dots$. Now let $\QQ(\TT) \ll  \P$ be the unique
martingale measure for $(S^{(n)}_{\tau_k})_{k=0,1,\dots}$ with respect
to $(\cF_{\tau_k})_{k=0,1,\dots}$ such that~\eqref{eq:111} holds also
for $\QQ(\TT)$-almost every $\omega$ with
$\tau_{k+1}(\omega)=1$. Hence, $\QQ(\TT)$ only gives probability to
the set of scenarios $\omega$ in which the terminal value
$S^{(n)}_1(\omega)$ is reached from the latest
$S^{(n)}_{\tau_k}(\omega)$ with $\tau_k(\omega)<1$ in a strictly
monotone way.

\begin{Lemma}\label{lem:1}
In the $n$-step binomial model~\eqref{eq:112} with fixed transaction
costs $\kappa>0$, the super-replication costs of a payoff $F=f(S^{(n)}_1)$
with $f$ convex on $(0,\infty)$ are
\begin{equation}
\VV^{\kappa}(f(S^{(n)}_1)) = \inf_{\TT \in \cT^{(n)}}
\E_{\QQ(\TT)}[f(S^{(n)}_1)+{\kappa} N(\TT)].\label{eq:7}
\end{equation}
\end{Lemma}
\begin{proof}
  Let us start by proving ``$\geq$'' in~\eqref{eq:7}. So take
  $x \in \RR$ and an admissible $(\TT,\HH)$ such that
  $x + G(\TT,\HH)_1 \geq f(S^{(n)}_1)$. By removing stopping points
  from $\TT=\cbr{\tau_k}_{k=0,1,\dots}$ if necessary we obtain a
  possibly coarser stopping system
  $\tilde{\TT}=\cbr{\tilde{\tau}_k}_{k=0,1,\dots}$ from our special
  class $\cT^{(n)}$ such that, under the unique martingale measure
  $\QQ(\tilde{\TT})$ for $(S^{(n)}_{\tilde{\tau}_k})_{k=0,1,\dots}$,
  we have $\tau_k=\tilde{\tau}_k$ almost surely for
  $k=0,1,\dots$. Therefore, we still have the super-replication
  property $x + G(\tilde{\TT},\HH)_1 \geq f(S^{(n)}_1)$
  $\QQ(\tilde{\TT})$-a.s. This allows us to conclude
  \begin{align*}
  x& = \E_{\QQ(\tilde{\TT})} [x+\sum_k
    h_k(S^{(n)}_{\tilde{\tau}_{k+1}}-S^{(n)}_{\tilde{\tau}_{k}})]\geq \E_{\QQ(\tilde{\TT})} [f(S^{(n)}_1)+\kappa N(\tilde{\TT}) ]
 \end{align*}
 as we wanted to show.

 We next establish ``$\leq$'' in~\eqref{eq:7}.  To this end, fix
 $\TT \in \cT^{(n)}$, put $\QQ \set \QQ(\TT)$, and denote
 $x:=\E_{\QQ}\left[f(S^{(n)}_1)+\kappa N(\TT)\right]$.  Observe that, under $\QQ$,
 the (frictionless) financial market with stock price process
 $(S^{(n)}_{\tau_k})_{k=0,1,\dots,n}$ is a binomial market and hence
 complete. The unique martingale measure is $\QQ$.  Thus, there exist
 measurable functions $\psi_k:\mathbb R^k_{+}\rightarrow\mathbb R$,
 $k=0,1,\dots,n$ such that
\begin{equation}\label{4.1}
x+\sum_{k=0}^{n-1} \psi_k\left(S^{(n)}_{\tau_1},...,S^{(n)}_{\tau_k}\right)
\left(S^{(n)}_{\tau_{k+1}}-S^{(n)}_{\tau_k}\right)=
f(S^{(n)}_1)+\kappa N(\TT) \quad \text{$\QQ$-a.s.}
\end{equation}
Let us now use these maps $\psi_k$, $k=0,1,\dots,n$, in order to
construct a super-replicating strategy for our $n$-step binomial
market with fixed transaction costs~$\kappa$. For this it will be
convenient to consider the obvious expansion of our binomial
model~\eqref{eq:112} from $[0,T]=[0,1]$ to all of $[0,\infty)$. Let
$\P^{(n)}$ still denote its locally equivalent martingale measure.
Use the mappings $\phi^\downarrow$ and $\phi^\uparrow$ associated with
the stopping system $\TT$ to define another system of stopping times
$\tilde{\TT}$ by $\tilde\tau_0 \set 0$ and, for $k=0,1,\dots$,
\begin{align}\label{4.new}
\tilde\tau_{k+1} \set &\min\cbr{t>\tilde\tau_k:\ln (\frac{S^{(n)}_t}{S^{(n)}_{\tilde\tau_k}})/\frac{\sigma}{\sqrt{n}}=\phi^\uparrow_k(S^{(n)}_{\tilde\tau_1},...,S^{(n)}_{\tilde\tau_k})}
\\&\wedge
\min\cbr{t>\tilde\tau_k: \ln (\frac{S^{(n)}_t}{S^{(n)}_{\tilde\tau_k}})/\frac{\sigma}{\sqrt{n}}=-\phi^\downarrow_k(S^{(n)}_{\tilde\tau_1},...,S^{(n)}_{\tilde\tau_k})}.\nonumber
\end{align}
Clearly, these successive two-sided level passage times
$\tilde\tau_1,...,\tilde\tau_n$ are finite $\P^{(n)}$-almost surely,
with $\tilde{\tau}_n \geq 1$. To obtain a strategy on $[0,1]$ we
truncate and consider the trading strategy $(\hat{\TT},\HH)$
intervening at times $\hat{\tau}_k \set \tilde{\tau}_k \wedge 1$
according to $\HH=(h_k)_{k=0,1,\dots}$ where
$$
h_k \set  \psi_k(S^{(n)}_{\tilde\tau_1},...,S^{(n)}_{\tilde\tau_k}), \quad k=0,1,\dots.
$$
In order to conclude our assertion, it is now sufficient to show that
$x+G(\hat{\TT},\HH)_1 \geq f(S^{(n)}_1)$ $\P^{(n)}$-a.s. In fact, we
will argue that
\begin{equation}\label{eq:800}
x+G(\tilde{\TT},\HH)_{\tilde{\tau}_n} \geq f(S^{(n)}_{\tilde{\tau}_n})
\text{ and } \tilde{\tau}_n \geq 1
\quad \text{$\P^{(n)}$-a.s.}
\end{equation}
which entails our assertion because
\begin{align*}
 x+G(\hat{\TT},\HH)_1 &=  x+G(\tilde{\TT},\HH)_1 \geq
 \E_{\P^{(n)}}[x+G(\tilde{\TT},\HH)_{\tilde{\tau}_n}\;|\;\cF^{(n)}_1]\\
& \geq  \E_{\P^{(n)}}[f(S^{(n)}_{\tilde{\tau}_n}) \;|\;\cF^{(n)}_1] \geq f(S^{(n)}_1).
\end{align*}
Here, the first estimate holds because $G(\tilde{\TT},\HH)$ is a
super-martingale under $\P^{(n)}$, the second estimate is due
to~\eqref{eq:800}, and the final one is due to Jensen's inequality for
the convex function $f$ and the $\P^{(n)}$-martingale $S^{(n)}$ (which
is uniformly bounded up to time $\tilde{\tau}_n$).

It remains to prove~\eqref{eq:800}. For this observe that by
construction both sides of this inequality are functionals of
$(S^{(n)}_{\tilde{\tau}_k})_{k=0,1,\dots}$. Moreover, this process is
a binomial martingale under $\P^{(n)}$ with exactly the same
jump characteristics as $(S^{(n)}_{\tau_k})_{k=0,1,\dots}$ under
$\QQ$ and, therefore,
\[
\Law((S^{(n)}_{\tilde{\tau}_k})_{k=0,1,\dots} \;|\;
\P^{(n)})=\Law((S^{(n)}_{{\tau}_k})_{k=0,1,\dots}\;|\;\QQ).
\]
As a consequence, \eqref{eq:800} is immediate from~\eqref{4.1}.
\end{proof}

For later use let us also note the following lemma which illustrates
the trade-off to be struck in our dual description of the
super-replication problem: For a convex payoff,
$\E_{\QQ(\TT)}[f(S^{(n)}_1)]$ may decrease when we add stops to $\TT$
while of course any added stop will let the number of interventions
$N(\TT)$ increase.

\begin{Lemma}\label{lem:2}
  If $\TT' \in \cT^{(n)}$ is a refinement of $\TT \in\cT^{(n)}$ in the
  sense that for any $\tau_k$ from $\TT$ we have
 $$
\tau_{k} = \max \cbr{\tau'_{k'} \in \TT' \;|\; \tau'_{k'} \leq \tau_k},
 $$
then for any convex payoff profile
  $f:(0,\infty)\to \RR$ we have
$$
\E_{\QQ(\TT')}[f(S^{(n)}_1)] \leq \E_{\QQ(\TT)}[f(S^{(n)}_1)].
$$
\end{Lemma}
\begin{proof}
  The measure $\QQ(\TT)$ is a martingale
  measure for $(S^{(n)}_{\tau_k})_{k=0,1,\dots}$ that is absolutely
  continuous with respect to $\P$ and which attains the frictionless
  super-replication price of the convex payoff $f(S^{(n)}_1)$ when trading
  is allowed only at times contained in $\TT$. Obviously, refining
  $\TT$ to $\TT' \in \cT^{(n)}$ offers more flexility to find
  super-replication strategies and thus cannot lead to a higher
  super-replication price.
\end{proof}

\subsection{Proof of the upper bound for super-replication
  prices}\label{sec:lowerbound}

In this section we will prove that ``$\geq$'' holds in our
formula~\eqref{eq:5} for the scaling limit. More precisely, we will
establish
\begin{equation}
  \label{eq:23}
  \liminf_n \VV^{\kappa/n}(f(S_1^{(n)})) \geq \inf_{\sigma \leq \nu \in \cA^W}
   \E^W\left[f(S^{(\nu)}_1)+\kappa \int_0^1 g(\nu_t^2/\sigma^2)\,dt\right]
\end{equation}
Without loss of generality (by passing to a sub--sequence) we assume that the limit
$\lim_n \VV^{\kappa/n}(f(S_1^{(n)}))$ exists in $[0,\infty]$.

By Lemma~\ref{lem:1}, we can find, for $n=1.2,\dots$, stopping systems
$\TT_0^{(n)} \in \cT^{(n)}$ such that
$$
 \VV^{\kappa/n}(f(S_1^{(n)})) \geq
 \E_{\QQ(\TT^{(n)}_0)}[f(S^{(n)}_1)+\frac{\kappa}{n} N(\TT^{(n)}_0)]-\frac{1}{n}.
$$
Hence, the $\liminf$ in~\eqref{eq:23} can be estimated if we get an
understanding, as $n \uparrow \infty$, of the joint law of
$S^{(n)}_1$ and $N(\TT^{(n)}_0)$ under $\QQ(\TT^{(n)}_0)$. While
tightness of this sequence of laws is not obvious, it can be
established for a suitable refinement of $\TT^{(n)}_0$ using an
argument which we adapt from Kusuoka~\cite{K}. To this end, fix
$m \in \cbr{1,2,\dots}$ and refine $\TT_0^{(n)}$ if necessary in such
a way that at most $m$ steps are taken between any two stopping
times. This gives us a stopping system
$\TT^{(n)} =\cbr{\tau^{(n)}_k}_{k=0,1,\dots}\in \cT^{(n)}$ with
$ N(\TT^{(n)}_0) \geq N(\TT^{(n)})-[(n-1)/m]$ and
\begin{equation}
  \label{eq:20}
  \tau^{(n)}_{k+1}-\tau^{(n)}_k \leq \frac{m}{n} \text{ on } \cbr{\tau^{(n)}_{k+1}<1}.
\end{equation}

In light of Lemma~\ref{lem:2}, we can now conclude that $\QQ^{(n)}
\set \QQ(\TT^{(n)})$ satisfies
\begin{equation}
  \label{eq:21}
 \VV^{\kappa/n}(f(S_1^{(n)})) \geq
 \E_{\QQ^{(n)}}[f(S^{(n)}_1)+\frac{\kappa}{n} N(\TT^{(n)})]-\frac{1}{n}-\frac{\kappa}{n}[(n-1)/m].
\end{equation}
Hence~\eqref{eq:23} will be established upon letting
$m\uparrow \infty$ once we can show that the
$\liminf_{n \uparrow \infty}$ of the expectations in~\eqref{eq:21} is
not smaller than the right-hand side of~\eqref{eq:23} for each
$m=1,2,\dots$. This will be accomplished using Kusuoka's tightness
argument for which we consider the processes $M^{(n)}$, $n=1,2,\dots$,
given by
\begin{align}\label{eq:301}
\begin{split}
M^{(n)}_{1}& \set S^{(n)}_{1},\\
M^{(n)}_{t} &\set S^{(n)}_{\tau^{(n)}_{k+1}} \text{ for } t \in
[\tau^{(n)}_k+1/n,\tau^{(n)}_{k+1}+1/n) \cap [0,1), \quad k=0,1,\dots.
\end{split}
\end{align}
Observe that $M^{(n)}$ is a version of the $\QQ^{(n)}$-martingale with
terminal value $S^{n}_1$:
$$
M^{(n)}_t = \E_{\QQ^{(n)}}[S^{(n)}_1 \;|\; \cF^{(n)}_t], \quad t
\in [0,1], \quad \text{ $\QQ^{(n)}$-a.s. }
$$

\begin{Lemma}\label{lem:3}
  Suppose $\TT^{(n)} \in \cT^{(n)}$, $n=1,2,\dots$, are partitions of
  $[0,1]$ such that~\eqref{eq:20} holds $\QQ^{(n)}$-almost surely
  where $\QQ^{(n)}=\QQ(\TT^{(n)})$. Then the sequence of distributions
  $(\Law(S^{(n)} \;|\; \QQ^{(n)}))_{n=1,2,\dots}$ is tight on the
  Skorohod space $\DD[0,1]$. Any weak accumulation point is the law of
  a strictly positive continuous martingale $M$ (in its own
  filtration) under some probability measure $\hat{\P}$ such that
 \begin{equation}\label{eq:500}
  \E_{\hat{\P}}[\max_{t \in [0,1]}( M_t)^p] \leq \sup_{n=1,2,\dots}
  \E_{\QQ^{(n)}}[\max_{t \in [0,1]} (M^{(n)}_t)^p]<\infty \text{ for any
  } p \geq 0.
  \end{equation}
  Moreover, the stochastic logarithm $L$ of $M/s_0$, i.e., the
  continuous local martingale $L$ such that $M=s_0\cE(L)$, has
  quadratic variation $\langle L \rangle$ absolutely continuous with
  respect to Lebesgue measure with density
  $\nu^2 \set d\langle L\rangle/dt \geq \sigma^2$.

 In addition, along a suitable subsequence, we have the weak convergence
\[
  \Law(S^{(n)}, \int_0^\cdot \alpha^{(n)}_{s} \,ds \;|\;
  \QQ^{(n)}) \to
  \Law(M,\int_0^\cdot \frac{1}{2}\left(\nu^2_t/\sigma^2-1\right)\,dt
  \;|\;\hat{\P}), \quad n \uparrow \infty,
\]
on $\DD[0,1] \times \DD[0,1]$ where
  \begin{equation}\label{eq:300}
     \alpha^{(n)}_t \set \frac{\sqrt{n}}{\sigma}|M^{(n)}_t-S^{(n)}_t|/S^{(n)}_t, \quad t \in [0,1]
 \end{equation}
 with $M^{(n)}$ given by~\eqref{eq:301}.
\end{Lemma}
\begin{proof}
  From (\ref{eq:20}) it follows that $\alpha^{(n)}_t$ is $\QQ^{(n)}$
  a.s. uniformly bounded (in $n$ and $t$), and so the tightness of
  $(\Law(S^{(n)} \;|\; \QQ^{(n)}))_{n=1,2,\dots}$ and the estimate
  (\ref{eq:500}) follow from Propositions 4.8 and 4.27 in
  \cite{K}. The second part of the lemma follows from Lemma 7.1 in
  \cite{DS}.
\end{proof}

By Skorohod's representation theorem, we can find processes
$\hat{S}^{(n)}$, $\hat{M}^{(n)}$, $\hat{\alpha}^{(n)}$, $n=1,2,\dots$,
on a common probability space $(\hat{\Omega},\hat{\cF},\hat{\P})$ which
have for each $n=1,2,\dots$ the same joint law as their
counterparts $(S^{(n)}, M^{(n)}, \alpha^{(n)})$ under $\QQ^{(n)}$ and
which are such that
$(\hat{S}^{(n)}, \hat{M}^{(n)}, \int_0^\cdot
\hat{\alpha}^{(n)}_u\,du)$ converges $\hat{\P}$-almost surely
uniformly in time to
$(\hat{M},\hat{M}, \int_0^\cdot
\frac{\hat{\nu}^2_t-\sigma^2}{2\sigma}\,dt)$ where
$\hat{M}=s_0 \cE(\hat{L})$ is a continuous $\hat{\P}$-martingale with
finite moments of arbitrary order and $\hat{\nu}^2$ is the density of
the quadratic variation of its stochastic logarithm $\hat{L}$ with
respect to Lebesgue meausre.

Moreover, for any $n=1,2,\dots$, we can define a system
$\hat{\TT}^{(n)}$ of stopping times $\hat{\tau}^{(n)}_k$,
$k=0,1,\dots$, for the filtration generated by $\hat{S}^{(n)}$ such
that also the joint $\hat{\P}$-law of these with
$(\hat{S}^{(n)},\hat{M}^{(n)})$ coincides with the joint law under
$\QQ^{(n)}$ of the stopping times $\tau^{(n)}_k$, $k=0,1,\dots$, with
$(S^{(n)},M^{(n)})$. In particular, we conclude from~\eqref{eq:20}
that
\[
\hat{\tau}^{(n)}_{k+1}-\hat{\tau}^{(n)}_{k} \leq \frac{m}{n}  \text{ $\hat{\P}$-a.s.}
\]
From~\eqref{eq:112} and~\eqref{eq:301}, we thus get the Taylor expansion
\[
\hat{\alpha}^{(n)}_t =
n\hat{\tau}^{(n)}_{k+1}-[nt]+O(m^2/\sqrt{n}) \text{
  for } t \in [\hat{\tau}^{(n)}_k+1/n,
\hat{\tau}^{(n)}_{k+1}+1/n) \cap [0,1] \text{ $\hat{\P}$-a.s.},
\]
where the absolute values of the $O(m^2/\sqrt{n})$-terms are uniformly in time and in
$\hat{\P}$-a.e. scenario less than or equal to $m^2/\sqrt{n}$.

The last observations allow us to apply Lemma~\ref{lem:A1} below (with
$b(t)\set \hat{\nu}^2_t/\sigma^2$) to get the estimate
\begin{equation}
\liminf_n \frac{N(\hat{\TT}^{(n)})}{n} \geq \int_0^1 g(\hat{\nu}^2_t/\sigma^2)\,dt
\text{ $\hat{\P}$-a.e.}\label{eq:306}
\end{equation}
where $g$ is the linearly interpolating function defined in
Theorem~\ref{thm:2}.

Taking $\liminf_n$ in~\eqref{eq:21} now gives
\begin{align*}
  \liminf_n  \VV^{\kappa/n}(f(S_1^{(n)}))
&\geq \liminf_n \E_{\QQ^{(n)}}[f(S^{(n)}_1)
+\kappa  N(\TT^{(n)})/n]-\frac{\kappa}{m}\\
&=\liminf_n \E_{\hat{\P}}[f(\hat{S}^{(n)}_1)
+\kappa N(\hat{\TT}^{(n)})/n]-\frac{\kappa}{m}\\
&\geq \E_{\hat{\P}}[f(\hat{M}_1)
+\kappa \int_0^1 g(\hat{\nu}^2_t/\sigma^2)\,dt] -\frac{\kappa}{m}
\end{align*}
where the final step is due to Fatou's lemma
and~\eqref{eq:306}. Applying a randomization technique similar to
Lemma~7.2 in \cite{DS} and letting $m\uparrow \infty$ now proves our
assertion~\eqref{eq:23}.

\begin{Lemma}
  \label{lem:A1}
    For $n=1,2,\dots$, let $\TT^{(n)}=\{0=t^{(n)}_0 \leq t^{(n)}_1
    \leq \dots \leq
    t^{(n)}_n=1\}$ be deterministic partitions of $[0,1]$ such that $n
    t^{(n)}_k \in \{0,1,\dots\}$ and $t^{(n)}_{k+1}-t^{(n)}_{k} \leq
    m/n$ for $k=0,1,\dots,n-1$. Suppose the functions
\[
a^{(n)}(t) \set nt^{(n)}_{k+1}-[nt], \quad t^{(n)}_k<t
\leq t^{(n)}_{k+1} \text{ for } k=0,1,\dots,n-1,
\]
satisfy
  \begin{equation}
\int_0^\cdot a^{(n)}(t) \,dt \to \int_0^\cdot \frac{1}{2}(b(t)-1) \,dt \text{
  uniformly on $[0,1]$}\label{eq:802}
\end{equation}
for some $b \in L^1([0,1],dt)$. Then we have
\[
 \liminf_{n} \frac{N(\TT^{(n)})}{n} \geq \int_0^1 g(b(t)) \,dt
\]
where $g$ is the linearly interpolating function defined in
Theorem~\ref{thm:2}.
\end{Lemma}
\begin{proof}
Without loss of generality (by passing to a sub--sequence) we assume that
$\lim_{n\rightarrow\infty}N(\TT^{(n)})/n$ exists.
For any $n$ introduce the function
$b_n:[0,1]\rightarrow [1,\infty)$ by $b_n(T)=0$ and
\begin{equation*}
b_n(t)=n(t^{(n)}_{k+1}-t^{(n)}_k), \ \ t^{(n)}_k\leq t<t^{(n)}_{k+1}, \ \ k=0,1,...,n-1,
\end{equation*}
for which we notice that
\begin{equation}\label{eq:611}
\frac{N(\TT^{(n)})}{n}=\int_{0}^1 \frac{1}{b_n(t)} dt.
\end{equation}
Simple calculations yield
\begin{align*}
\int_{t^{(n)}_k}^{t^{(n)}_{k+1}}&\left[\frac{1}{2}(b_n(t)-1)-a_n(t)\right]dt=0.
\end{align*}
This together with~\eqref{eq:802} and
the fact that $n(t^{(n)}_{k+1}-t^{(n)}_k)$ is bounded uniformly in $k$
and $n$ gives
\begin{equation}\label{eq:622}
\int_{0}^t b_n(u) du\rightarrow \int_{0}^t b(u) du \text{ uniformly in
  } t\in [0,1].
\end{equation}

The Komlos Lemma (see Lemma A 1.1 in \cite{DS1}) implies that there
exists a sequence of functions $\tilde{b}_n \in \conv(b_n,b_{n+1},...)$,
$n=1,2,\dots$, such that $\tilde{b}_n$ converges Lebesgue-almost everywhere
to a function $\tilde{b}$.  In fact, $\tilde{b}=b$ a.e. since by dominated convergence
and~\eqref{eq:622} we get
$$
\int_{0}^t \tilde{b}(u)du=\lim_{n\rightarrow\infty}\int_{0}^t \tilde{b}_n(u)du=
\lim_{n\rightarrow\infty}\int_{0}^t b_n(u)=\int_{0}^t b(u) du \text{
  for any } t\in[0,1].
$$
Finally, from~\eqref{eq:611}, the fact that the function $g$ is convex and continuous with
$g(b_n)=\frac{1}{b_n}$ (as $b_n$ is integer valued) we obtain
\begin{align*}
\lim_{n}\frac{N(\TT^{(n)})}{n}=
\lim_{n}\int_{0}^1 g(b_n(t))dt\geq
\lim_{n}\int_{0}^1 g(\tilde{b}_n(t))=\int_{0}^1 g(b(t))dt
\end{align*}
and the result follows.
\end{proof}

\subsection{Proof of the lower bound for super-replication prices}
\label{sec:upperbound}

In this section we will establish ``$\leq$'' for our
formula~\eqref{eq:5} for the scaling limit of super-replication
prices. More precisely, we will prove
\begin{equation}
  \label{eq:423}
  \limsup_n \VV^{\kappa/n}(f(S_1^{(n)})) \leq
   \E^W\left[f(S^{(\nu)}_1)+\kappa \int_0^1 g(\nu_t^2/\sigma^2)\,dt\right]
\end{equation}
for any volatility processs $\nu \geq \sigma$ in $\cA^W$ on some filtered
probability space $(\Omega^W, \cF^W,  (\cF^W_t), \P^W)$ supporting a
Brownian motion $W$ as considered in Theorem~\ref{thm:2}. In fact, it
suffices to show this for piecewise constant $\nu$:

\begin{Lemma}\label{lem:4}
  For any $\nu \in \cA^W$ and any $\epsilon>0$, there is
  $\tilde{\nu} \in \cA^W$ of the simple form
\begin{equation}
  \label{eq:401}
  \tilde{\nu}_t = \sum_{j=0}^J \sigma
  \sqrt{\rho_j(S^{(\tilde{\nu})}_{t_0},\dots, S^{(\tilde{\nu})}_{t_j})} 1_{(t_j,t_{j+1}]}(t)
\end{equation}
for some times $0=t_0<t_1<\dots<t_J=1$ and continuous bounded
functions $\rho_j:\RR^{j+1} \to [1,\infty)$ such that
\[
\left|\E^W\left[f(S^{(\nu)}_1)+\kappa \int_0^1 g\left(\nu_t^2/\sigma^2\right)\,dt\right]-\E^W\left[f(S^{(\tilde{\nu})}_1)+\kappa \int_0^1 g\left(\tilde{\nu}_t^2/\sigma^2\right)\,dt\right]\right|<\epsilon.
\]
\end{Lemma}
\begin{proof}
  Let $\nu \in \cA^W$ and let $C$ be a constant such that $\nu\leq C$
  a.s. Using similar density arguments as in Lemma 3.4 in \cite{BDP}
  (for $d=1$) we get that there exists a sequence $\nu^{(n)}$,
  $n=1,2,\dots$, such that $\nu^{(n)}\leq C$ is of the simple form
  given by~\eqref{eq:401} and $\nu^{(n)}\rightarrow\nu$ $\P^W\otimes dt$-a.e. This
  together with the uniform integrability (due to $\nu^{(n)} \leq C$)
  of the sequence
  $f(S^{(\nu^{(n)})}_1)+\kappa \int_0^1 g((\nu^{(n)}_t)^2/\sigma^2)\,dt$,
  $n=1,2,\dots$, implies the assertion.
\end{proof}
In the proof of~\eqref{eq:423} we can assume without loss of
generality that $\lim_n \VV^{\kappa/n}(f(S_1^{(n)}))$ exists.  The
duality result in Lemma~\ref{lem:1} suggests to construct a sequence
of stopping systems $\TT^{(n)} \in \cT^{(n)}$ with respect to
$(\cF^{(n)}_t)$, $n=1,2,\dots$, such that under the associated
measures $\QQ^{(n)} \set \QQ(\TT^{(n)})$ the processes $S^{(n)}$
of~\eqref{eq:112} converge in law to $S^{(\nu)}$. This will be done
next.

To fix ideas, let us first focus on the initial period
$[t_0,t_1)=[0,t_1)$ where we wish  to obtain the constant
$\nu^2_0=\sigma^2 \rho_0 \in [0,\infty)$ as the limiting local
variance. Inspection of the argument in the previous section suggests
that for $\rho_0 \in \cbr{1,2,\dots}$ this can be accomplished by
stopping any $\rho_0$ consecutive upwards or downwards steps (and not
stop before the end in scenarios without this monotonicity property). For $\nu^2_0$
between natural multiples of $\sigma^2$, though, we have to mix
stopping after $[\rho_0]$ steps and after $[\rho_0]+1$ steps in just
the right proportions.  For instance, if we want to obtain
asymptotically the local variance $1.5 \sigma^2$ (i.e. $\rho_0=1.5$),
we just alternate between stopping after $[\rho_0]=1$ steps and after
$[\rho_0]+1=2$ steps in the same direction (and again do not stop before
the end in all scenarios which are incompatible with this).

In general, the following construction will work: For $j=0,\dots,J$,
we subdivide the time interval $[[n t_j]/n,[nt_{j+1}]/n) $ into
$[nt_{j+1}]-[nt_j]\approx n(t_{j+1}-t_j)=O(n)$ periods of length
$1/n$. These $O(n)$ periods can be covered by
$\sqrt{n(t_{j+1}-t_j)}=O(\sqrt{n})$ blocks of the same number
$\sqrt{n(t_{j+1}-t_j)}= O(\sqrt{n})$ of successive time points. Denote by
$$\rho^{(n)}_j \set \rho_j(S^{(n)}_{[n t_0]/n}, \dots, S^{(n)}_{[n
  t_j]/n})$$ a proxy for the multiple of $\sigma^2$ we want to
implement asymptotically as local variance over the interval
$[t_j,t_{j+1})$. Take $\lambda^{(n)}_j$ to be the unique solution
$\lambda \in (0,1]$ of
$$
 \rho^{(n)}_j = \lambda [\rho^{(n)}_j]+(1-\lambda)([\rho^{(n)}_j]+1).
$$
In each of the above $O(\sqrt{n})$ blocks of length $O(\sqrt{n})$, we
will first stop every time after $[\rho^{(n)}_j]$ steps have been made
by the binomial model consecutively in the same direction (i.e. all
upwards or all downwards) and we will not stop at all before reaching
the time horizon $T=1$ in scenarios where different directions are
taken in this period. This continues until we have covered a fraction
of $\lambda^{(n)}_j$ of the present block's $O(\sqrt{n})$ periods. For
the remaining fraction $1-\lambda^{(n)}_j$ of periods in this block,
we will proceed similarly but with a rhythm of stopping every
$[\rho^{(n)}_j]+1$ steps instead of $[\rho^{(n)}_j]$. After that we
repeat this procedure for all of the $O(\sqrt{n})$ blocks we separated
the interval $[[n t_j]/n,[nt_{j+1}]/n)$ into in the beginning. Then we
proceed similarly with the next interval
$[[n t_{j+1}]/n,[nt_{j+2}]/n)$ until all of these intervals are
treated.

Let us next analyze the asymptotic transaction costs and variance
which this procedure entails. We can do this separately on each of the
intervals $[t_j,t_{j+1})$, $j=0,\dots,,J$. So fix such a $j$ and let
$n_1$ and $n_2$ denote the number of times where we stop every
$[\rho^{(n)}_j]$ and $[\rho^{(n)}_j]+1$ binomial steps,
respectively. Then we have
$$
n_1 [\rho^{(n)}_j]+n_2([\rho^{(n)}_j]+1)=\sqrt{n(t_{j+1}-t_j)}+O(1)
$$
and, in order to obtain the right asymptotic variance for $M^{(n)}$
constructed from the thus obtained $
\tau^{(n)}_j$s as in~\eqref{eq:301}, we want to have at the same time that
 $$
n_1 [\rho^{(n)}_j]^2+n_2([\rho^{(n)}_j]+1)^2 =\rho^{(n)}_j \sqrt{n(t_{j+1}-t_j)}+O(1).
$$
We conclude
\begin{align*}
\frac{n_1}{\sqrt{n(t_{j+1}-t_j)}}&=\frac{1+[\rho^{(n)}_j]-\rho^{(n)}_j}{[\rho^{(n)}_j]}+O(1/\sqrt{n}),
\\
\frac{n_2}{\sqrt{n(t_{j+1}-t_j)}}&=\frac{\rho^{(n)}_j-[\rho^{(n)}_j]}{1+[\rho^{(n)}_j]}+O(1/\sqrt{n}),
\end{align*}
and the fraction of periods covered in $[\rho^{(n)}_j]$ steps, respectively, is the desired
$$
\lambda^{(n)}_j = \frac{n_1[\rho^{(n)}_j]}{\sqrt{n(t_{j+1}-t_j)}}+O(1/\sqrt{n})=
1+[\rho^{(n)}_j]-\rho^{(n)}_j+O(1/\sqrt{n})
$$
The transaction costs on this block are equal to $ (n_1+n_2) \kappa/n$,
and so we conclude that the transaction costs on the whole interval
$[[n t_j]/n,[nt_{j+1}]/n)$ amount to
\begin{equation}\label{eq:901}
\sqrt{n(t_{j+1}-t_j)}(n_1+n_2)\kappa/n=\kappa (t_{j+1}-t_j)g(\rho^{(n)}_j)+O(1/\sqrt{n}).
\end{equation}
Furthermore, we get that for any $t\in (t_j,t_{j+1})$ the process
$\alpha^{(n)}$ as in~\eqref{eq:300} satisfies
\begin{eqnarray}\label{estimate}
&\int_{t_j}^t\alpha^{(n)}_s ds=O(1/\sqrt n)+\frac{t-t_j}{\sqrt{n(t_{j+1}-t_j)}}\left(0+1+2+...+[\rho^{(n)}_j]-1\right)n_1\\
&+\frac{t-t_j}{\sqrt{n(t_{j+1}-t_j)}}\left(0+1+2+...+[\rho^{(n)}_j]\right)n_2
=O(1/\sqrt n)+\frac{\rho^{(n)}_j-1}{2}(t-t_j).\nonumber
\end{eqnarray}
Having constructed for $n=1,2,\dots$ a system of stopping times
$\TT^{(n)} = \cbr{\tau^{(n)}_k}\in \cT^{(n)}$, we can let
$\QQ^{(n)} \set \QQ(\TT^{(n)})$ denote the associated martingale
measure for $(S^{(n)}_{\tau^{(n)}_k})_{k=0,1,\dots}$. Observe that
along with the functions $\rho_j$ also the $\rho_j^{(n)}$ are bounded
uniformly, say by a constant $m \in \cbr{1,2,\dots}$. As a
consequence, the increments between any two successive intervention
times are bounded by $m/n$ $\QQ^{(n)}$-almost surely as
in~\eqref{eq:20}. We can thus invoke Lemma~\ref{lem:3} to conclude
that, possibly along a subsequence again denoted by $n$, we have the
weak convergence
\[
  \Law(S^{(n)}, \int_0^\cdot \alpha^{(n)}_{s} \,ds \;|\;
  \QQ^{(n)}) \to
  \Law(M,\int_0^\cdot \frac{1}{2}\left(\hat{\nu}^2_t/\sigma^2-1\right)\,dt
  \;|\;\hat{\P}), \quad n \uparrow \infty,
\]
on $\DD[0,1] \times \DD[0,1]$ for some $\hat{\nu} \geq \sigma$ with
$\hat{\nu}^2=d\langle L\rangle/dt$ for the stochastic logarithm $L$ of
$M=s_0 \cE(L)$.
In fact, $M$ and $\hat{\nu}$ are just copies, respectively, of our
original $S^{(\nu)}$ and $\nu$; see Lemma~\ref{lem:6} below.

Just as after Lemma~\ref{lem:3} in the previous section, we now use
Skorohod's representation theorem to see that without loss of
generality we can assume to have $\hat{S}^{(n)}$, $\hat{M}^{(n)}$ and
$\hat{\alpha}^{(n)}$ on $(\hat{\Omega},\hat{\cF},\hat{\P})$ which have,
for each $n=1,2,\dots$, the same joint law as their counterparts
$(S^{(n)}, M^{(n)}, \alpha^{(n)})$ under $\QQ^{(n)}$ and which are
such that, as $n \uparrow \infty$,
\begin{equation}
  \label{eq:601}
  (\hat{S}^{(n)}, \hat{M}^{(n)}, \int_0^\cdot
\hat{\alpha}^{(n)}_u\,du) \to (M,M, \int_0^\cdot \frac{1}{2}(\hat{\nu}^2_t/\sigma^2-1)\,dt)
\end{equation}
uniformly in time $\hat{\P}$-almost surely, where
$\hat{\nu}^2 \set d\langle L\rangle/dt$ for the stochastic logarithm $L$
of $M$. For $n=1,2,\dots$, we can also reconstruct from
$\hat{S}^{(n)}$ a system of stopping times $\hat{\TT}^{(n)}$ for the
filtration generated by $\hat{S}^{(n)}$ which corresponds to our
$\TT^{(n)}$ constructed above.
From (\ref{estimate})--(\ref{eq:601}) and the fact that
the functions $\rho_j$, $j=1,...,J$ are continuous it follows that
\begin{equation}\label{specialform}
 \hat{\nu}_t = \sum_{j=0}^J \sigma
  \sqrt{\rho_j(M_{t_0},\dots, M_{t_j})} 1_{(t_j,t_{j+1}]}(t) \quad \P
  \otimes dt\text{-a.e.}
\end{equation}

The proof of~\eqref{eq:423} is now completed by arguing that
\begin{align*}
    \VV^{\kappa/n}(f(S^{(n)}_1) &\leq
                                  \E_{\QQ^{(n)}}[f(S^{(n)}_1)+\frac{\kappa}{n}
                                  N(\TT^{(n)})]\\
& =  \E_{\hat{\P}}[f(\hat{S}^{(n)}_1)]+\E_{\hat{\P}}[\kappa N(\hat{\TT}^{(n)})/n]\\
&\to \E_{\hat{\P}}[f(M_1)]+\E_{\hat{\P}}[\kappa\int_0^1 g(\hat{\nu}_t^2/\sigma^2)\,dt]\\
&=\E^W\left[f(S^{(\nu)}_1)+\kappa \int_0^1 g(\nu_t^2/\sigma^2)\,dt\right].
\end{align*}
Here the estimate in the first line is immediate from
Lemma~\ref{lem:1} and the first identity is due to our Skorohod
representation. The convergence
$\E_{\hat{\P}}[f(\hat{S}^{(n)}_1)] \to \E_{\hat{\P}}[f(M_1)]$ is due
to dominated convergence since uniform integrability follows from the
polynomial growth of $f$ and~\eqref{eq:500}; the convergence of the
other expectations also follows by dominated convergence since
$N(\hat{\TT}^{(n)})/n \in [0,1]$, $n=1,2,\dots$, and
since~\eqref{eq:901} in conjunction with~\eqref{specialform} yields
$\hat{\P}$-a.s. convergence of the costs $\kappa N(\hat{\TT}^{(n)})/n$
to $\kappa\int_0^1g(\hat{\nu}_t^2/\sigma^2)\,dt$. The final identity
is immediate from Lemma~\ref{lem:6} below.

\begin{Lemma}\label{lem:6} We have
\begin{equation}
  \label{eq:502}
  \Law(S^{(\nu)} \;|\; \P^W) = \Law(M \;|\; \hat{\P}).
\end{equation}
\end{Lemma}
\begin{proof}
Let us prove by induction that, for any
$j=0,1,...,J$, the distribution of $M|_{[0,t_j]}$ is equal to the
distribution of $S^{(\nu)}|_{[0,t_j]}$.  For $j=0$ the statement is
trivial. Assume that the statement is correct for $j$.  Define the
stochastic process
$$B_t=\frac{1}{\sigma \sqrt{\rho_j(M_{t_0},\dots, M_{t_j})}}\int_{t_j}^{t+t_j}\frac{d M_u}{M_u},\ \ t\in[0,t_{j+1}-t_j].$$
From the Levy Theorem and
(\ref{specialform}) it follows that $B$ is
a Brownian motion on $[0,t_{j+1}-t_j]$ independent of $M|_{[0,t_j]}$. Clearly, for
$t\in [t_j, t_{j+1}]$,
\begin{equation}\label{5.mar1}
M_t=M_{t_j}\exp\left(\sigma \sqrt{\rho_j(M_{t_0},..,M_{t_j})}B_{t-t_j}-
\sigma^2\rho_j(M_{t_0},..,M_{t_j})(t-t_j)/2\right).
\end{equation}
On the other hand, for $t\in [t_j, t_{j+1}]$,
\begin{equation}\label{5.mar2}
S^{(\nu)}_{t}=S^{(\nu)}_{t_j}\exp\left(\sigma \sqrt{\rho_j(S^{(\nu)}_{t_0},..,S^{(\nu)}_{t_j})}\hat B_{t-t_j}-
\sigma^2\rho_j(S^{(\nu)}_{t_0},..,S^{(\nu)}_{t_j})(t-t_j)/2\right)
\end{equation}
where $\hat B_t=W_{t+t_j}-W_{t_j}$, $t\geq 0$ is a Brownian motion
independent of $S^{(\nu)}|_{[0,t_j]}$.
From~\eqref{5.mar1}--\eqref{5.mar2} and the induction assumption we
get that the distribution of $M|_{[0,t_{j+1}]}$ coincides with the
distribution of $S^{(\nu)}|_{[0,t_{j+1}]}$ as required. Hence, the
distribution of $M$ is the same as that of $S^{(\nu)}$.
\end{proof}


\bibliographystyle{spbasic}

\end{document}